\documentclass[review]{elsarticle}
\usepackage[a4paper,top=2cm,bottom=2cm,left=3.5cm,right=3.5cm]{geometry}

\usepackage{amsmath}
\usepackage{amsthm}
\usepackage{amssymb}
\usepackage{braket}
\usepackage{booktabs}
\usepackage{algorithm}
\usepackage{graphicx}
\usepackage{braket}
\usepackage{enumerate}
\usepackage[title]{appendix}
\usepackage[colorlinks=true, allcolors=blue]{hyperref}

\newtheorem{Problem}{Problem}
\newtheorem{Lemma}{Lemma}
\newtheorem{Theorem}{Theorem}

\newtheorem{Fact}{Fact}

\theoremstyle{definition}

\newcommand\F{\mathbb{F}}

\title{Unbounded quantum-classical separation in sample complexity for sphere center finding}

\begin{document}
\begin{frontmatter}
\author{Guanzhong Li$^{a}$}
\author{Lvzhou Li$^{a,b,}$\corref{mycorrespondingauthor}}%\footnote{L. Li is the corresponding author (lilvzh@mail.sysu.edu.cn)}}
\address{$^a$Institute of Quantum Computing and Software, School of Computer and Engineering, Sun Yat-sen University, Guangzhou 510006, China}
\address{$^b$Quantum Science Center of Guangdong-Hong Kong-Macao Greater Bay Area (Guangdong),  Shenzhen 518045, China}
% \cortext[email]{ligzh9@mail2.sysu.edu.cn}
\cortext[mycorrespondingauthor]{lilvzh@mail.sysu.edu.cn}

\begin{abstract}
Fast quantum algorithms can solve important computational problems more efficiently than classical algorithms.
However, little is known about whether quantum computing can speed up solving geometric problems.
This article explores quantum advantages for the problem of finding the center of a sphere in vector spaces over finite fields, given samples of random points on the sphere.
We prove that any classical algorithm for this task requires approximately as many samples as the dimension of the vector space,  by a reduction to an old and basic algebraic result---Warning's second theorem. On the other hand, we propose a quantum algorithm based on quantum walks that needs only a constant number of samples to find the center.
Thus, an unbounded quantum advantage is revealed for a natural and intuitive geometric problem, which highlights the power of quantum computing in solving geometric problems.
\end{abstract}

\end{frontmatter}

\section{Introduction}
A primary goal of quantum computing is to identify problems for which quantum algorithms can offer speedups over their classical counterparts.
Grover's algorithm~\cite{Grover} offers a quadratic speedup for unstructured database search problem.
Shor's algorithm~\cite{shor} is exponentially faster than the best {known} classical algorithm in factoring integers.
Recent studies have also found exponential speedups in specific problems such as simulating the coupled oscillators~\cite{coupled2023}, traversing a decorated welded-tree graph using adiabatic quantum computation~\cite{adiabatic2021}, a graph property testing problem in the adjacency list model~\cite{Gproperty2020}, a classification problem for quantum kernel methods based on discrete logarithm problem~\cite{Liu2021}, a specific NP search problem based on random black-box function~\cite{Verifiable2022}, and so on.
However, to the best of our knowledge, there has been less study of exploring quantum speedups on   geometric problems.
%no one has ever found a natural  geometric problem that exhibits sharp quantum speedups.

In this article, we consider a natural and intuitive  geometric problem of finding the center of a sphere given samples of random points on the sphere.
To illustrate this problem, we begin with an intuitive example: if we want to pin down a circle on a paper, we will need $3$ points ($A, B, C$) which are not on the same line;
and if we want to fix an inflatable balloon, we will need $4$ points ($A,B,C,D$) which are not on the same plane.
As such, $n+1$ points are required to determine a sphere in $\mathbb{R}^n$.

To accommodate the discrete nature of qubit-based quantum computer, we will replace the continuous vector space $\mathbb{R}^n$ with $\F_p^n$, the vector space over a finite field $\F_p$ with $p$ being a prime number.
Define the length of a vector $x=(x_1,\cdots,x_n)\in \F_p^n$ by $l(x) := \sum_{j=1}^n x_j^2$. 
Then a sphere with radius $r\in \F_p \setminus\{0\}$ and center $t\in \F_p^n$ is denoted by $S_r+t = \{x+t \mid x\in S_r\}$, where
\begin{equation}\label{eq:S_r_def}
    S_r := \{ x\mid x\in\F_p^n, l(x)=r \}.
\end{equation}
The problem is then formally described as follows.
\begin{Problem}\label{prob:main}
    Find the unknown center $t\in\F_p^n$ of the sphere $S_r+t$ ($r\neq 0$ is given) with as few samples of points on $S_r+t$ as possible.
\end{Problem}

It is natural to ask if the problem is well-defined, or more precisely, could different centers $t \neq t'$ result in the same sphere $S_r +t = S_r +t'$?
For example, when $n=2$ and $r=1$, the spheres centered at $t=(0,0)$ and $t'=(1,1)$ respectively are the same: $S_1+t =S_1+t' = \{(1,0),(0,1)\}$.
However, we will show in Appendix~\ref{app:unique} that when the prime $p$ is odd and the dimension $n \geq 2$, different centers lead to different spheres, and thus Problem~\ref{prob:main} is well-defined.

In analogy to sampling a point uniformly random from $S_r+t$, each usage of the  state $\ket{S_r+t} := \frac{1}{\sqrt{s_r}} \sum_{v \in S_r+t} \ket{v}$, where $s_r := |S_r|$, is regarded as a quantum sample or qSample (a term coined by Aharonov and Ta-Shma~\cite{qSample}).
Such a way of comparing the number of classical samples versus coherent quantum samples has been a common measure of complexity known as sample complexity in quantum learning theory~\cite{survey_learning_theory}, and quantum samples can be more powerful than classical ones when they are drawn from the uniform distribution.
For example, it was shown by Bshouty and Jackson~\cite{BJ99} that Disjunctive Normal Form (DNF) formulas are learnable in polynomial time from quantum samples under the uniform distribution, while classically the best upper bound is quasi-polynomial time~\cite{Ver90}.
Another example is the learning of $k$-juntas (functions that depend on at most $k$ of the $n$ input bits), where Atıcı and Servedio~\cite{Atıcı2007} showed that $(\log n)$-juntas can be learned with quantum samples under the uniform distribution in $\mathrm{poly}(n)$ time, while classically the best algorithm runs in quasi-polynomial time~\cite{MOSSEL04}.
For an arbitrary and unknown distribution, Arunachalam and de Wolf~\cite{Arunachalam_18} showed that quantum sample complexity has exactly the same asymptotic scaling as classical learning complexity.

In this article, we obtain the following theorem.
%To the best of our knowledge, it seems to be the first quantum advantage over a geometric problem using quantum walks.

\begin{Theorem}\label{thm:main}
    There is a quantum algorithm that solves Problem~\ref{prob:main} with bounded error, and uses $O(\log p)$ \footnote{The standard asymptotic notations $O(\cdot)$ and $\Omega(\cdot)$ are used.  We say the complexity is  $O(f(n))$ (resp. $\Omega(f(n))$), if for large enough $n$, it is at most (resp. at least) $cf(n)$ for some constant $c$.} samples.
    However, any classical algorithm  solving Problem~\ref{prob:main} with bounded error requires $\Omega(n)$ samples.
\end{Theorem}

Note that when $p$ is fixed to be a constant, the quantum algorithm needs constant $O(1)$ samples, while any classical algorithm requires $\Omega(n)$ samples.
Thus, we obtain an unbounded quantum advantage.
{It is noteworthy that we do not claim a quantum time advantage, nor do we claim to be the first to prove an unbounded quantum-classical separation in sample complexity.}

To prove the classical lower bound  $\Omega(n)$,
we find unexpectedly that it can be reduced to a basic algebraic theorem proved around 1935 by Warning~\cite{warning1935}, which gives a lower bound on the number of zeros of multi-variable polynomials over finite fields.
We then design the quantum algorithm based on  continuous-time quantum walks (CTQW) on the Euclidean graph on $\F_p^n$, where two points $x,y$ are connected if and only if $l(x-y) = r$.
Our algorithm is inspired by earlier work  by Childs, Schulman, and Vazirani~\cite{nonlinear}, but it requires a fine-grained analysis of the success probability so as to adapt to our problem.

Quantum walks are an analogy to classical random walks, and have become a widely adopted paradigm to design quantum algorithms for various problems, such as spatial search~\cite{spatial_2004,coins2005,unified,robust,universal}, element distinctness~\cite{Ambainis07,eedp}, matrix product verification~\cite{BuhrmanS06}, triangle finding~\cite{MagniezSS07,LI2025105295}, group commutativity~\cite{MagniezN07}, the welded-tree problem~\cite{CCD03,multi2023,welded2024}, and the hidden flat of centers problem~\cite{nonlinear}.
There are two types of quantum walks: the CTQW and the discrete-time quantum walk (DTQW).
CTQW is relatively simple, and mainly involves simulating a Hamiltonian $H$ that encodes the structure of the graph.
DTQW is more diverse, ranging from the earliest and simplest coined quantum walk~\cite{AmbainisBNVW01,AmbainisKV01} to various Markov chain based frameworks~\cite{Szegedy_03,MagniezNRS11,belovs2013quantum,KroviMOR16,unified}.
%Although quantum walks have been employed to solve various problems, they have not been used to address any geometric problem as far as we know. Our result represents the first step in this direction, and we hope it will inspire further explorations of quantum computing in this area.
%The quantum algorithm proposed in this article is based on CTQW and is different from previous ones with sharp speedups featuring quantum Fourier transform, and thus our result may inspire the discovery of new quantum algorithms with sharp speedups.

\section{Classical lower bound}
In the introduction section we have intuitively shown that $n+1$ points are required to determine a sphere in $\mathbb{R}^n$.
Here we will rigorously prove that any randomized algorithm requires $\Omega(n)$ samples of points on the sphere $S_r+t$ to determine the center $t$ with high probability.

Let $\mathcal{Y} := \F_p^n$ be the \textit{solution space} that contains all the possible centers $t$.
Suppose the maximum number of sample points on the sphere $S_r+t$ that a randomized algorithm can get in Problem~\ref{prob:main} is $h$, which will be used in Theorem~\ref{thm:classical_lower} to upper bound the success probability.
Let $\mathcal{X} := (\F_p^n)^h$ be the \textit{sample space}, so that for any possibly obtained sample combination $x = (x_1,\cdots,x_h) \in (S_r+t)^h$ where $t \in \mathcal{Y}$ is arbitrarily fixed, we have $x \in \mathcal{X}$.
A randomized algorithm for Problem~\ref{prob:main} using at most $h$ sample points can be characterized by a function:
\begin{equation}\label{eq:rho}
	\rho: \mathcal{X} \to \mathcal{D}(\mathcal{Y}),
\end{equation}
where $\mathcal{D}(\mathcal{Y})$ consists of all the probability distributions on the set $\mathcal{Y}$.
In other words, for a possibly obtained sample combination $x\in \mathcal{X}$, the randomized algorithm $\rho$ outputs an answer $y\in \mathcal{Y}$ with probability $\rho(x,y) \in [0,1]$.

Denote by $X \sim U((S_r+y)^h)$ the fact that the random variable $X$ obeys the uniform distribution on $(S_r+y)^h$.
% Specifically, $\Pr(X=x) = {1}/s_r^h$ for any $x \in (S_r+y)^h$ and $\Pr(X=x) = 0$ for any $x\in (\F_p^n)^h\setminus (S_r+y)^h$.
We can now present Theorem~\ref{thm:classical_lower} that implies the classical lower bound.
% upper bounds the worst-case success probability of any randomized algorithm that solves Problem~\ref{prob:main} using at most $h \leq cn$ sample points.
\begin{Theorem}\label{thm:classical_lower}
Let $c$ be a constant such that $c\in(0,1/2)$.
Assume that $h \leq cn$, then for any randomized algorithm $\rho: \mathcal{X} \to \mathcal{D}(\mathcal{Y})$ using at most $h$ sample points, where $\mathcal{X} = (\F_p^n)^h$ and $\mathcal{Y} = \F_p^n$, we have:
\begin{equation}\label{eq:success_pr}
	\min_{y\in \mathcal{Y}} \mathop{\mathbb{E}}\limits_{X\sim U((S_r+y)^h)} [\rho(X,y)]
	\leq 1/p^{(1-2c)n}.
\end{equation}
\end{Theorem}
% Let us explain the meaning of Theorem~\ref{thm:classical_lower}.
% The expression $\mathbb{E}_{X\sim U((S_r+y)^h)} [\rho(X,y)]$ in Eq.~\eqref{eq:success_pr} calculates the average probability of a randomized algorithm $\rho$ that successfully recovers the correct center $y$, when the sample combination $x = (x_1,\cdots,x_h)$ are taken uniformly from $(S_r+y)^h$.
If we let $c=1/3$ in Theorem~\ref{thm:classical_lower}, then a classical algorithm $\rho$ that uses less than $h \leq n/3$ samples on $S_r+y$ can only succeed with probability at most $1/p^{n/3}$ in the worst case,
% (by the ``$\min_{y\in\mathcal{Y}}$" in Eq.~\eqref{eq:success_pr}),
which is exponentially small with $n$.
Therefore, if a classical algorithm wants to solve Problem~\ref{prob:main} with constant probability, $\Omega(n)$ samples are needed.

%\subsection{Proof of Theorem~\ref{thm:classical_lower}}
The rest of this section is devoted to the proof of Theorem~\ref{thm:classical_lower}. An outline of our proof is shown below:
\begin{enumerate}
    % \item We first give another definition of a randomized algorithm as a probability distribution on deterministic strategies (Eq.~\eqref{eq:p}), which will be shown  equivalent to the aforementioned definition as shown by Eq.~\eqref{eq:rho} (Lemma~\ref{lem:rho_to_p}), but is easier to deal with.

    \item We first present Yao's minimax principle specific to randomized algorithms for sample problems (Lemma~\ref{thm:minimax}), which converts randomness in randomized algorithms to randomness in the solution space $\mathcal{Y}$.
    % which can be derived from von Neumann's minimax theorem (Eq.~\eqref{eq:min_max_von} in Appendix~\ref{app:minimax}).
    % The principle allows us to convert randomness within the randomized algorithm to randomness in the solution space $\mathcal{Y} = \F_p^n$, or more precisely, to construct a probability distribution $\mu \in \mathcal{D}(\mathcal{Y})$ (the uniform distribution suffices) that is difficult for any deterministic strategy.

    \item We then show a lower bound of deterministic strategies, using Warning's second theorem (Lemma~\ref{lem:warning}).
    Specifically, we will show in Lemma~\ref{lem:m_x} that even after seeing $cn$ sample points on the sphere, there are still exponentially remaining centers consistent with these sample points.

    \item Finally, we finish the proof of Theorem~\ref{thm:classical_lower} by combining the above two results together.
    % lower bound for deterministic strategies with Yao's minimax principle for randomized algorithms.    
\end{enumerate}

\subsection{Yao's minimax principle specific to randomized algorithms for sample problems}

Recall from Problem~\ref{prob:main} that given a center $y\in \mathcal{Y}$, the sample combination $x$ are drawn from the uniform distribution $U((S_r+y)^h) \in \mathcal{D}(\mathcal{X})$.
More generally, the following function characterizes a sample problem:
\begin{equation}\label{eq:sigma}
	\sigma: \mathcal{Y} \to \mathcal{D}(\mathcal{X}).
\end{equation}
A deterministic strategy can be characterized by a function $f: \mathcal{X} \to \mathcal{Y}$, which can also be regarded as a vector $f \in \mathcal{Y}^{|\mathcal{X}|}$ since $\mathcal{X}$ and $\mathcal{Y}$ are both finite sets.

% Recall from Eq.~\eqref{eq:rho} that $\rho: \mathcal{X} \to \mathcal{D}(\mathcal{Y})$ defines a randomized algorithm, and recall from above Eq.~\eqref{eq:p} that a function $f: \mathcal{X} \to \mathcal{Y}$ or its vector form $f \in \mathcal{Y}^{|\mathcal{X}|}$ defines a deterministic strategy.

We can now present Lemma~\ref{thm:minimax}, which can be seen as Yao's minimax principle~\cite{Yao1977} specific to randomized algorithms for sample problems.
The principle allows us to convert randomness within the randomized algorithm to randomness in the solution space $\mathcal{Y} = \F_p^n$, or more precisely, to construct a probability distribution $\mu \in \mathcal{D}(\mathcal{Y})$ that is difficult for any deterministic strategy.

\begin{Lemma}\label{thm:minimax}
For a sample problem $\sigma: \mathcal{Y} \to \mathcal{D}(\mathcal{X})$ with corresponding solution space $\mathcal{Y}$ and sample space $\mathcal{X}$, we have:
\begin{equation}\label{eq:minimax_sample}
	\min_{\mu \in \mathcal{D}(\mathcal{Y})}
	\max_{f \in \mathcal{Y}^{|\mathcal{X}|}}
	\mathop{\mathbb{E}}\limits_{Y\sim \mu}
	[\Pr_{X \sim \sigma(Y)} (f(X)=Y)]
	= \max_{\rho\in [\mathcal{X} \to \mathcal{D}(\mathcal{Y})]}
	\min_{y\in \mathcal{Y}}
    \mathop{\mathbb{E}}\limits_{X\sim \sigma(y)} [\rho(X,y)].
\end{equation}
\end{Lemma}
Further explanation and the usage of Lemma~\ref{thm:minimax} are as follows: In order to obtain an upper bound on the worst-case success probability $\min_{y\in \mathcal{Y}} \mathop{\mathbb{E}}\limits_{X\sim \sigma(y)} [\rho(X,y)]$ of any random algorithm $\rho\in [\mathcal{X} \to \mathcal{D}(\mathcal{Y})]$, according to the outer $\min_{\mu \in \mathcal{D}(\mathcal{Y})}$ and $\max_{\rho\in [\mathcal{X} \to \mathcal{D}(\mathcal{Y})]}$ in Eq.~\eqref{eq:success_pr}, we only need to construct a probability distribution $\mu \in \mathcal{D}(\mathcal{Y})$ on the solution space, and find the upper bound on the expected success probability of any deterministic strategies $f \in \mathcal{Y}^{|\mathcal{X}|}$ under this probability distribution, i.e. $\max_{f \in \mathcal{Y}^{|\mathcal{X}|}} \mathop{\mathbb{E}}\limits_{Y\sim \mu} [\Pr_{X \sim \sigma(Y)} (f(X)=Y)]$.
For completeness, the proof of Lemma~\ref{thm:minimax} is presented in Appendix~\ref{app:minimax}.
% where we reduce it to von Neumann's minimax theorem~\cite{Neumann1928}.
% Note also that Lemma~\ref{thm:minimax} is inspired by Ref.~\cite[Section~2.2]{quantum_yao}, but they considered randomized algorithms for computing Boolean functions there, which is quite different from the randomized algorithms for sample problems considered here.
% $f: D \to \{0,1\}$ 
% $\rho: \mathcal{X} \to \mathcal{D}(\mathcal{Y})$
% $\sigma: \mathcal{Y} \to \mathcal{D}(\mathcal{X})$

\subsection{Lower bound on sample complexities of deterministic strategies}

Recall that a deterministic strategy for Problem~\ref{prob:main} with sample complexity $h$ can be characterized by a function $f: (\F_p^n)^h \to \F_p^n$ that outputs an answer $y \in \F_p^n$ when given a sample combination $x \in \bigcup_{y\in \F_p^n} (S_r+y)^h =: \mathcal{S} \subsetneq (\F_p^n)^h$, where $\mathcal{S} \subsetneq (\F_p^n)^h$ follows from the fact that the $h$ sample points cannot be taken from different spheres.

We now present Lemma~\ref{lem:m_x}, which implies that even after a deterministic strategy $f$ has seen $n/3$ sample points on the sphere, there are still $p^{n/3}$ exponentially remaining centers consistent with these sample points.
Therefore, for a deterministic strategy $f$ to succeed with certainty, its sample complexity needs to be greater than $n/3$.

\begin{Lemma}\label{lem:m_x}
Assume that $h \leq cn$, where constant $c\in(0,1/2)$.
For any sample combination $x \in \mathcal{S} = \bigcup_{y\in \F_p^n} (S_r+y)^h$, let $m_x$ denote the number of {$y'\in \F_p^n $} such that $x \in (S_r+y')^h$.
Then it holds that $m_x \geq p^{(1-2c)n}$.
% Consider any possible sample combination $x \in \mathcal{S} := \bigcup_{y\in \F_p^n} (S_r+y)^h$.
% Assume that there are $m_x$ many $\{y'\} \subseteq \F_p^n $ such that $x \in (S_r+y')^h$, then $m_x \geq p^{(1-2c)n}$.
% we have:
% \begin{equation}
%     m_x \geq p^{(1-2c)n}.
% \end{equation}
\end{Lemma}

To prove Lemma~\ref{lem:m_x}, we will use Lemma~\ref{lem:warning} in the following, also known as Warning's second theorem~\cite{warning2018,Clark2017,Clark2018}, attributed to Ewald Warning~\cite{warning1935}.
% \cite[Theorem 3.5]{Arithmetic}

\begin{Lemma}[Warning's second theorem]\label{lem:warning}
    Suppose $f_1,f_2,\cdots,f_s \in \F_p[y_1,\cdots,y_n]$ are multivariate polynomials over $\F_p$ with $d_i =\mathrm{deg}(f_i)$.
    Let $d:= d_1 + \cdots +d_s$ be the total degree and $V := \{ y=(y_1,\cdots,y_n)\in \F_p^n \mid f_i(y)=0, \forall\, 1\leq i \leq s \}$ be the set of common zeros of $f_i$.
    Assume $n>d$ and $V\neq \emptyset$.
    Then $|V| \geq p^{n-d}$.
\end{Lemma}

\begin{proof}[Proof of Lemma~\ref{lem:m_x}]
Write {the} sample combination $x \in (S_r+y)^h$ as a set $X+y$ of points on the sphere $S_r+y$, where $X := \{ (x_{i1},\cdots,x_{in}) \mid i=1,\ldots,h \} \subseteq S_r$.
Assume that there are $m_x$ many possible centers {$ y' \in \F_p^n $} such that $X+y \subseteq S_r+y'$, i.e. consistent with the sample points.
Denote $\bar{y} := y'-y$, then the inclusion $X +y \subseteq S_r +y'$ becomes $X\subseteq S_r+\bar{y}$, which results in the following system of polynomial equations over finite fields about variables $\bar{y}=(\bar{y}_1,\cdots,\bar{y}_n)$:
\begin{equation}
    f_i := \sum_{j=1}^{n} (x_{ij} -\bar{y}_j)^2 -r=0, \ 1\leq i\leq h.
\end{equation}
Since $d_i =\mathrm{deg}(f_i) = 2$ for all $i$, and $\bar{y} = \vec{0}$ is a solution (recall that $X \subseteq S_r$), the above equations have at least $p^{n-2h} \geq p^{(1-2c)n}$ common roots $\bar{y}$ by Lemma~\ref{lem:warning}.
Thus there are $m_x \geq p^{(1-2c)n}$ possible centers $y'=\bar{y}+y$ such that $X +y \subseteq S_r +y'$.
\end{proof}

\subsection{Finishing the proof of Theorem~\ref{thm:classical_lower}}

To prove that $\min_{y\in \F_p^n} \mathbb{E}_{X\sim U((S_r+y)^h)} [\rho(X,y)] \leq p^{(1-2c)n}$ in Theorem~\ref{thm:classical_lower}, we simply let $\mu \in \mathcal{D}(\F_p^n)$ be the uniform distribution, i.e. $\mu = U(\F_p^n)$, in Eq.~\eqref{eq:minimax_sample} of Lemma~\ref{thm:minimax}.
Thus it suffices to show the following inequality:
\begin{equation}\label{eq:E_Y_U}
	\mathbb{E}_{Y\sim U(\F_p^n)} [\Pr_{X \sim U((S_r+Y)^h)} (f(X)=Y)]
	\leq 1/p^{(1-2c)n}.
\end{equation}
Let $\delta(A) \in \{1,0\}$ be the boolean function that indicates whether event $A$ happens or not.
Denote by $\mathcal{S} = \bigcup_{y\in \F_p^n} (S_r+y)^h$ the set of all possible sample combinations as in Lemma~\ref{lem:m_x}.
Expanding the LHS of Eq.~\eqref{eq:E_Y_U}, we have:
\begin{align}
	&\frac{1}{p^n} \sum_{y\in \F_p^n}
	\frac{1}{s_r^h} \sum_{x \in (S_r+y)^h}
	\delta(f(x)=y) \label{eq:double_sum_line_1} \\
	&=\frac{1}{p^n s_r^h} \sum_{x\in \mathcal{S}}
	\sum_{y'\in \F_p^n \,{\rm s.t.}\, x \in (S_r+y')^h} \delta(f(x)=y') \label{eq:double_sum_line_2}\\
	% &\leq\frac{1}{p^n s_r^h} \sum_{x\in \mathcal{S}} \label{eq:double_sum_line_3}\\
	&\leq \frac{|\mathcal{S}| }{p^n s_r^h} \label{eq:double_sum_line_3}\\
	&\leq 1/p^{(1-2c)n}, \label{eq:double_sum_line_4}
\end{align}
where the inequality in Eq.~\eqref{eq:double_sum_line_3} uses the fact that $f(x)$ as a function of $x \in (\F_p^n)^h$ can only take one value in $\F_p^n$, possibly contained in the $m_x$ many {$y' \in \F_p^n $} such that $x \in (S_r+y')^h$.
The inequality in Eq.~\eqref{eq:double_sum_line_4} follows from the fact that $p^n s_r^h = \sum_{x\in \mathcal{S}} m_x \geq |\mathcal{S}| p^{(1-2c)n}$, where the inequality follows from Lemma~\ref{lem:m_x}.

\section{Quantum algorithm}
The quantum algorithm for Problem~\ref{prob:main} is concise and consists of only four steps as shown below.
Our algorithm is inspired by earlier work  by Childs, Schulman, and Vazirani~\cite{nonlinear}, but it requires a fine-grained analysis of the success probability so as to adapt to our problem.
The main idea is to use CTQW on the Euclidean graph $G$  to move amplitude from the sphere to its center.
The Euclidean graph $G$ has vertex set $\F_p^n$, and two vertices $x,x' \in F_p^n$ are connected by an undirected edge if and only if $l(x-x') = r$.

The Hamiltonian $\bar{A}$ of the CTQW on the Euclidean graph $G$ is approximately the adjacency matrix $A$ of the Euclidean graph $G$, but with $A$'s largest eigenvalue $\lambda_0 = s_r$ replaced by $0$.
Specifically, $\bar{A} := A -\lambda_0 \ket{\widetilde{0}} \bra{\widetilde{0}}$, where $\ket{\tilde{0}}$ is the corresponding eigenvector of eigenvalue $\lambda_0$ and it is the equal superposition of all points in $\F_p^n$ (see~\cite[Proposition 2]{analogue} for the spectral decomposition of $A$).
Note that $A$ is symmetric since the graph $G$ is undirected, so $\bar{A}$ is a valid Hamiltonian. We will see later from numerical simulation in Appendix~\ref{sec:numerical} that $A$ itself as the Hamiltonian is good enough.

\begin{enumerate}
    \item Prepare the quantum sample $\ket{S_r+t} = \frac{1}{\sqrt{s_r}} \sum_{v \in S_r+t} \ket{v}$.

    \item Apply a CTQW $e^{i\bar{A} T_0}$ to $\ket{S_r+t}$ for time $T_0 = 1/\sqrt{p^{n-1} \log p}$, and then measure in the computational basis obtaining a point in $\F_p^n$.    

    \item Repeat the above two steps for $O(\log p)$ times.
    
    \item Output the point with the highest frequency.
\end{enumerate}

The following lemma lower bounds the success probability of step~2.
It can be seen as a fine-grained version of~\cite[Lemma 4]{nonlinear}.
\begin{Lemma}\label{lem:CTQW_prob}
    The final state $\ket{\psi_{T_0}} := e^{i\bar{A}T_0} \ket{S_r+t}$ of CTQW on the Euclidean graph $G$ has the following properties:
    \begin{equation}
        \left|\braket{x | \psi_{T_0}}\right|
    \begin{cases}
        \leq O(p^{-(n-1)/2}) &\mathrm{for}\, x\neq t, \\
        \geq \frac{h(p,n)}{\sqrt{\log p}} &\mathrm{for}\, x=t,
    \end{cases} 
    \end{equation}
    where the function $h(p,n)$ is monotonically increasing in both $p$ and $n$, and $h(127,8) \geq 0.0016$.
\end{Lemma}

From Lemma~\ref{lem:CTQW_prob}, the probability to obtain the center $t$ is $\Omega(1/\log p)$ in step~2, while the probability to obtain any point other than $t$ is exponentially small.
Step~3 then guarantees the center $t$ to be found with high probability.
A more detailed analysis can be found in Appendix~\ref{app:p_q_lower}.
Overall, the quantum algorithm needs $O(\log p)$ samples.

\begin{proof}[Proof of Lemma~\ref{lem:CTQW_prob}]
   As shown in Appendix~\ref{app:CTQW}, the CTQW $e^{i \bar{A} T_0} = \sum_{k=0}^{\infty} \frac{(iT_0)^k}{k!} \bar{A}^k$ has a $(p+1)$-dimensional invariant subspace $\mathcal{H}_0$ spanned by the following orthonormal basis:
    \begin{equation}\label{eq:B_0}
        \mathcal{B}_0 = \{ \ket{t}, \ket{S_0+t}, \ket{S_1+t},\cdots,\ket{S_{p-1}+t} \},
    \end{equation}
    where $S_0 := \{v\in\F_p^n: l(v)=0\} \setminus \{\vec{0}\}$ (see Eq.~\eqref{eq:S_r_def} for $S_r$ with $r\neq 0$).

The following equation~\cite[Theorem 1]{analogue} shows that $s_r := |S_r|\approx p^{n-1}$ for all $r \in \F_p$.

\begin{align}\label{eq:sphere_size}
% s_r &:= |S_r| \nonumber\\
s_r &=
\begin{cases}
p^{n-1} +\chi((-1)^{\frac{n-1}{2}}r) p^{\frac{n-1}{2}} & n\, \mathrm{odd}, r\neq 0, \\
p^{n-1} -\chi((-1)^{\frac{n}{2}}) p^{\frac{n-2}{2}} &n\, \mathrm{even}, r\neq 0, \\
p^{n-1}-1 & n\, \mathrm{odd}, r=0, \\
p^{n-1} + \chi((-1)^{\frac{n}{2}})(p-1)p^{\frac{n-2}{2}}-1 &n\, \mathrm{even}, r=0,
\end{cases}
\end{align}
where $\chi(a) \in\{0,1,-1\}$ indicates whether $a$ is zero, or the square of some element in $\F_p$, or otherwise.
    
    We first consider $x\neq t$.
    Since $\ket{\psi_{T_0}} \in \mathrm{span}\, \mathcal{B}_0$, we have:
    \begin{align}
        \left| \braket{x | \psi_{T_0}} \right| & \leq \max_{r'\in\F_p} \left| \braket{x | S_{r'}+t } \right| \\
        &= \max_{r'\in\F_p} 1/\sqrt{s_{r'}} = O(p^{-(n-1)/2}).
    \end{align}
    We then consider the lower bound of $\left|\braket{t | \psi_{T_0}} \right|$ as follows:
    \begin{align}
        &\left|\braket{t | \psi_{T_0}} \right|
        = \left|\bra{t} e^{i \bar{A} T_0} \ket{S_r +t}\right| 
        =\left|\bra{t} \sum_{k=0}^{\infty} \frac{(iT_0)^k}{k!} \bar{A}^k \ket{S_r +t}\right| \\
        &\geq T_0 \bra{t} \big(A -\lambda_0 \ket{\widetilde{0}} \bra{\widetilde{0}} \big) \ket{S_r +t} -\sum_{k=2}^{\infty} \frac{T_0^k}{k!} \left|\bra{t} \bar{A}^k \ket{S_r+t}\right| \label{eq:lwb_line_2}\\
        &\geq T_0 \left( \sqrt{s_r} - s_r  \frac{{1}}{\sqrt{p^n}} \frac{{\sqrt{s_r}}}{\sqrt{p^n}} \right)
        -\sum_{k=2}^{\infty} \frac{(2T_0\sqrt{p^{n-1}})^k}{k!} \label{eq:lwb_line_3}\\
        &= \sqrt{s_r T_0^2} -  \frac{s_r}{p^n} \sqrt{s_r T_0^2} -(e^{2T_0\sqrt{p^{n-1}}} -1 -2T_0\sqrt{p^{n-1}}).
    \end{align}
    We have used $\braket{t|S_r+t}=0$ and the triangle inequality $|x+y| \geq |x| -|y|$ and the fact that $\bar{A} = A -\lambda_0 \ket{\widetilde{0}} \bra{\widetilde{0}}$ is a real matrix in Formula~\eqref{eq:lwb_line_2}. 
    Formula~\eqref{eq:lwb_line_3} is because the adjacency matrix $A$ maps $\ket{t}$ to $\sqrt{s_r} \ket{S_r+t}$,
    {$\lambda_0=s_r$, $\ket{\tilde{0}}$ is the uniform superposition over $\F_q^n$,}
    and $|\bra{t} \bar{A}^k \ket{S_r+t}| \leq \|\bar{A}^k\| \leq (2\sqrt{p^{n-1}})^k$, where the second upper bound follows from the fact that the spectral radius of $\bar{A}$ is less than $2\sqrt{p^{n-1}}$~\cite[Theorem 3]{analogue}.

    From Eq.~\eqref{eq:sphere_size} we know $p^{n-1} -p^{(n-1)/2} \leq s_r \leq p^{n-1} +p^{(n-1)/2}$.
    Recall that $T_0 = 1/\sqrt{p^{n-1} \log p}$.
    Thus we have $\frac{1}{\log p}(1-\frac{1}{p^{(n-1)/2}}) \leq s_r\,T_0^2 \leq \frac{1}{\log p}(1+\frac{1}{p^{(n-1)/2}})$, and $s_r/p^n \leq \frac{1}{p} (1+\frac{1}{p^{(n-1)/2}})$.
    Using basic inequalities $\sqrt{a-\varepsilon} \geq \sqrt{a} -\sqrt{\varepsilon}$ and $\sqrt{a+\varepsilon} \leq \sqrt{a} +\sqrt{\varepsilon}$, we can continue to calculate the lower bound of $\left|\braket{t | \psi_{T_0}} \right|$ as follows:
    \begin{align}
    & \left|\braket{t | \psi_{T_0}} \right| 
        \geq \frac{1}{\sqrt{\log p}} (1 -\frac{1}{p^{(n-1)/4}})
        \nonumber \\ &\quad
        -\frac{1}{\sqrt{\log p}}(1 +\frac{1}{p^{(n-1)/4}}) (\frac{1}{p} +\frac{1}{p^{(n+1)/2}})
        % \nonumber \\ &\quad
        -(e^{2/\sqrt{\log p}} -1 -\frac{2}{\sqrt{\log p}}) \\
        &\geq \frac{1}{\sqrt{\log p}} (1 -\frac{1}{p} -\frac{4}{p^{(n-1)/4}} -g(\sqrt{\log p})) \label{eq:lwb2_line_2}\\
        &:= \frac{1}{\sqrt{\log p}} h(p,n)
    \end{align}
    In Formula~\eqref{eq:lwb2_line_2}, $g(x) := x e^{2/x} -x -2$ is monotonically decreasing when $x > 0$, and $g(\sqrt{\log p}) < 1$ when $p \geq 127$.
    Thus $h(p,n)$ is monotonically increasing in both $p$ and $n$, and it can be verified that $h(127,8) \geq 0.0016$.
\end{proof}

\section{Conclusion}
In summary, we have found  unbounded  quantum advantages for a natural  geometric problem of finding the center of a sphere in $\F_p^n$ given samples on the sphere.
While any classical bounded-error algorithm requires $\Omega(n)$ samples on the sphere, the quantum algorithm based on CTQW  needs only $O(1)$ sample{s}. We hope that quantum algorithms will be able to demonstrate their advantages in solving more geometric problems in the future.

\section*{Acknowledgments}
This work was supported by the National Key Research and Development Program of China (Grant No.2024YFB4504004), the National Natural Science Foundation of China (Grant No. 92465202, 62272492, 12447107),  the Guangdong Provincial Quantum Science Strategic Initiative (Grant No. GDZX2303007, GDZX2403001), the Guangzhou Science and Technology Program (Grant No. 2024A04J4892).

\bibliographystyle{quantum}
\bibliography{reference}

\begin{appendices}

\section{Numerical simulation}\label{sec:numerical}
To illustrate that $A$ is good enough as the Hamiltonian, we consider the simplest cases where the size of the finite field is $p=3$ and the radius of the sphere is $r=1$.
The invariant subspace $\mathcal{H}_0 = \mathrm{span}\, \mathcal{B}_0$ of CTQW  $e^{iAT_0}$ is then $4$-dimensional:
\begin{equation}
    \mathcal{B}_0 = \{\ket{t}, \ket{S_0+t}, \ket{S_1+t}, \ket{S_2+t} \}.
\end{equation}

We first consider the case where $n=5$.
Using Lemma~\ref{lem:invariant} in Appendix~\ref{app:CTQW}, we have the following matrix expression of $A$ on the basis $\mathcal{B}_0$.
\begin{equation}\label{eq:M_A}
    M_A = \begin{bmatrix}
0 & 0 & 3\sqrt{10} & 0 \\
0 & 27 & 24\sqrt{2} & 9\sqrt{10} \\
3\sqrt{10} & 24\sqrt{2} & 33 & 12\sqrt{5} \\
0 & 9\sqrt{10} & 12\sqrt{5} & 30
\end{bmatrix}.
\end{equation}
Denote by $p(T_1) := \left|\bra{t} e^{iA T_1\cdot T_0} \ket{S_r+t} \right|^2$ the success probability of reaching the target state after time $T_1 \cdot T_0$, then $p(T_1) = \left|\bra{0} e^{iM_A T_1\cdot T_0} \ket{2} \right|^2$ for the case where $p=3,r=1,n=5$.
By numerical calculation, we obtain $p(T_1)$ with $T_1 \in [0,10]$ as shown in Fig.~\ref{fig:P_t}.
It shows that the success probability of the CTQW oscillates periodically with its evolution time, and the earliest time to achieve the maximum success probability is around $T_1 = 1.5$.

\begin{figure}[hbtp]
\centering
\includegraphics[width=0.7\linewidth]{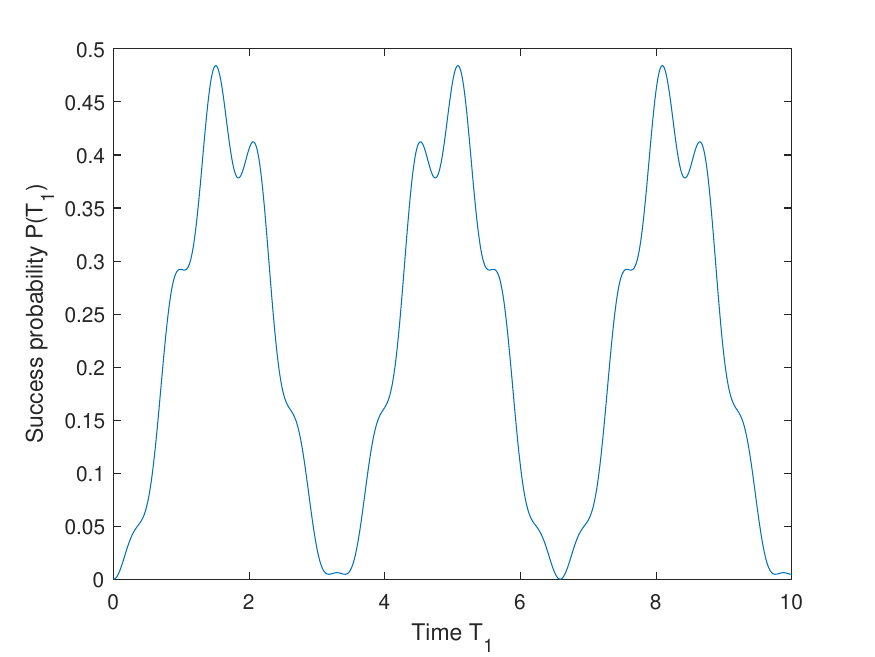}
\caption{\label{fig:P_t} The success probability $p(T_1)$ for the case where $p=3,r=1,n=5$.}
\end{figure}

We then consider the cases where $n=3\sim 12$.
Fig.~\ref{fig:P_t_max} shows that as the dimension $n$ varies, the optimal evolution time $T_\mathrm{max} := \mathop{\arg\max}\limits_{T_1 \in [0,3]} p(T_1)$ lies in $[1,2]$,
and the corresponding maximum success probability $p(T_\mathrm{max}) \geq 0.4$.

\begin{figure}[hbtp]
\centering
\includegraphics[width=0.7\linewidth]{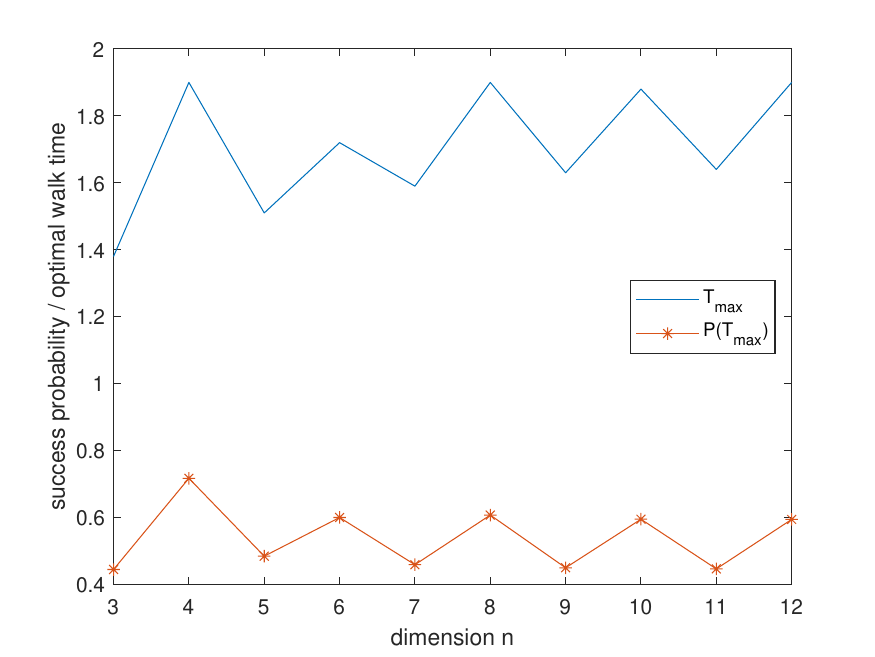}
\caption{\label{fig:P_t_max} Optimal CTQW time $T_\mathrm{max}$ and corresponding success probability $p(T_\mathrm{max})$. Here, $p=3, r=1$, and $n=3\sim 12$.}
\end{figure}

\section{Uniqueness of the center of a sphere}\label{app:unique}

Recall that the sphere with center $t\in \F_p^n$ and radius $r \in \F_p \setminus\{0\} =: \F_p^*$ is denoted by $S_r+t := \{ x+t \mid x \in S_r \}$, where $S_r := \{ x \in \F_p^n \mid l(x) =r \}$ and $l(x) := \sum_{j=1}^{n} x_j^2$.
In the rest of this section, we will prove the following theorem.

\begin{Theorem}\label{thm:unique}
    For two different points $t \neq t'$ in $\F_p^n$, where $p$ is an odd prime number and dimension $n\geq 2$, the spheres centered at them both with radius $r \in \F_p^*$ are different, namely $S_r+t \neq S_r+t'$.
\end{Theorem}

Let $\bar{t} = t'-t$, then it is equivalent to showing that $S_r \neq S_r +\bar{t}$ if $\bar{t} \neq 0$.
We will prove this by contradiction.
Assume the opposite that $S_r +\bar{t} = S_r$, then $x_i + \bar{t} \in S_r +\bar{t} =S_r$ for any $x_i \in S_r$.
From the definition of $S_r$, we have $l(x_i+\bar{t}) = r$, which gives the following $s_r = |S_r|$ equations:
\begin{equation}\label{eq:x_ij_t_square}
	\sum_{j=1}^{n} (x_{ij}+\bar{t}_j)^2 = r, \quad \forall i\in\{1,\dots,s_r\}.
\end{equation}

In order to simplify Eq.~\eqref{eq:x_ij_t_square}, note that $\sum_{j=1}^{n} x_{ij}^2 =r$ by the fact that $x_i \in S_r$, and thus Eq.~\eqref{eq:x_ij_t_square} is simplified to $\sum_{j=1}^{n} 2x_{ij} \bar{t}_j +\sum_{j=1}^{n} {\bar{t}_j}^2 =0$.
Denote by $a$ the multiplicative inverse of $2$ in the finite field $\F_p$, i.e. $2a=1$ (note that $a$ exists and $a\neq 0$, because $2$ is coprime with the odd prime $p$), then Eq.~\eqref{eq:x_ij_t_square} is further simplified to:
\begin{equation}\label{eq:a_x_ij}
	a \sum_{j=1}^{n} {\bar{t}_j}^2 + \sum_{j=1}^{n} x_{ij} \bar{t}_j =0, \quad \forall i\in\{1,\dots,s_r\}.
\end{equation}

The above Eq.~\eqref{eq:a_x_ij} can be regarded as a system of linear equations in $n+1$ variables $(\sum_{j=1}^{n} {\bar{t}_j}^2, \bar{t}_1,\cdots, \bar{t}_n)$, with coefficient matrix $A$ whose first column consists entirely of $a$, and whose rest $n$ columns and $s_r$ rows are such that each row is a point $x_i = (x_{i1},\cdots,x_{in})$ in $S_r$.
We will later prove the following Fact~\ref{fact:linear_independent}, from which the system of linear equations shown in Eq.~\eqref{eq:a_x_ij} has only the all zero solution, implying $\bar{t} =0$.
This contradicts $\bar{t} \neq 0$, which indicates that the initial assumption $S_r +\bar{t} = S_r$ is wrong, and thus Theorem~\ref{thm:unique} is proved.

\begin{Fact}\label{fact:linear_independent}
    The coefficient matrix $A$ of the system of linear equations shown in Eq.~\eqref{eq:a_x_ij} contains a square submatrix $B$ with size $n+1$ which is full-rank.
\end{Fact}

To prove Fact~\ref{fact:linear_independent}, we first consider the simple case where the radius $r=1$ and dimension $n=3$.
The proof of the general case is basically the same as this simple case.
Consider the unit sphere $S_1$ in $\F_p^3$, which obviously contains the following $4 = n+1$ points: $(1,0,0)$, $(0,1,0)$, $(0,0,1)$ and $(-1,0,0)$.
So we get a square submatrix with size $n+1$ of the coefficient matrix $A$:
\begin{equation}\label{eq:matrix_B_1}
	B=
\left[\begin{array}{c|ccc}
a & 1 & 0 & 0 \\
a & 0 & 1 & 0 \\
a & 0 & 0 & 1 \\
 \hline
a & -1 & 0 & 0
\end{array}\right].
\end{equation}
To prove that $B$ is full-rank, it suffices to show that its determinant $|B| \neq 0$.
To calculate $|B|$, we perform the following row operations on $B$: subtract the last row from the first $3=n$ rows, and arrive at:
\begin{align}\label{eq:det_B_1}
	|B| &=
\left|\begin{array}{c|ccc}
0 & 2 & 0 & 0 \\
0 & 1 & 1 & 0 \\
0 & 1 & 0 & 1 \\
 \hline
a & -1 & 0 & 0
\end{array}\right|
\\& = (-1)^{4+1} a\cdot 2
\\& \neq 0.
\end{align}

It is easy to see that the above argument also works for dimensions $n \geq 2$, as long as the radius $r=1$.
Note that the argument relies on the fact that the equation $x^2 = 1$ in $\F_p$ has two solutions $1$ and $-1$.
However, for general radii $r\in \F_p^*$, the equation $x^2 =r$ in $\F_p$ may not have any solution.
For example, squaring all the elements in $\F_3$ results in $\{0,1\}$ excluding $2$, meaning $x^2 = 2$ has no solution in $\F_3$.
In this case, $2$ is called a quadratic non-residue modular $3$.
To address this issue, we will use the following Fact~\ref{fact:sum_of_squares}.

\begin{Fact}\label{fact:sum_of_squares}
    For any $r \in \F_p^*$, there exist $x \in \F_p^*$ and $y \in \F_p$ such that $x^2 +y^2 =r$.
\end{Fact}

\begin{proof}
Consider the function $f: x \mapsto x^2$ on $\F_p$.
We first show that the range of $f$ (denoted by $Q$, whose elements are called quadratic residues modular $p$) is of size $\frac{p+1}{2}$.

Since $f(0) =0$, and $f(\F_p^*) \subseteq \F_p^*$ (if $x\neq 0$ and $x^2 =0$, then multiplying both sides by the inverse of $x$ yields $x=0$, a contradiction), it suffices to show that $f$ is two-to-one on $\F_p^*$.
To do so, it suffices to show that for any $b \in \F_p^*$, the equation $x^2 =b$ in $\F_p$ either has no solutions or has exactly two solutions.
Suppose there is a solution $c \neq 0$, then the equation $x^2 =b$ becomes $x^2 =c^2$, which is $(x-c)(x+c)=0$.
This equation has at most two solutions: $x =c$ and $x=-c$.
All that remains is to prove that these two solutions are different, namely $c\neq -c$.
Assuming $c = -c$, then $2c = 0$, multiplying both sides by $a$, i.e. the multiplicative inverse of $2$, gives $c = 0$, which contradicts with $c \neq 0$.

For any $r \in \F_p$, the size of the sets $Q$ and $r-Q$ are both $\frac{p+1}{2}$, and since they are both subsets of $\F_p$, their intersection is non-empty.
This implies that there exist $b_1,b_2 \in Q$ such that $b_1 =r-b_2$.
Assume $b_1 =x^2$ and $b_2 =y^2$, then $x^2 +y^2 =r$.
If $r \neq 0$, then one of $x,y$ must be non-zero (otherwise $r=0$, a contradiction), so Fact~\ref{fact:sum_of_squares} is proven.
\end{proof}

According to Fact~\ref{fact:sum_of_squares}, there exist $x \in \F_p^*$ and $y \in \F_p$ such that $x^2 +y^2 =r$ for $r \in \F_p^*$.
We now prove Fact~\ref{fact:linear_independent} for the two cases $y=0$ and $y\neq 0$.

\textbf{Case I}: $y = 0$.
Denote by $e_k$ the row vector with $1$ in the $k$-th column and $0$ in the rest of the columns.
For $k\in\{1,\dots,n\}$, let $v_k := a e_1 +x e_{k+1}$, and let ${v}_{n+1} := ae_1 -xe_2$.
Since $x^2 =r$ now, the $(n+1)$-sized square matrix $B$ with $(v_1,\dots,v_{n+1})$ as its row vectors is a submatrix of the coefficient matrix $A$ of the system of linear equations shown in Eq.~\eqref{eq:a_x_ij}.
More specifically, taking $n=3$ as an example, the square matrix $B$ only needs to replace $1$ in Eq.~\eqref{eq:matrix_B_1} with $x$.
Therefore, similar to the calculation process shown in Eq.~\eqref{eq:det_B_1}, we know that the determinant of the square matrix $B$ is $|B| = (-1)^{n+2}\cdot a \cdot 2x^n$.
Since $a \neq 0$ and $x\neq 0$, we have $|B| \neq 0$, and thus Fact~\ref{fact:linear_independent} holds.

\textbf{Case II}: $y\neq 0$.
For $k \in \{1,\dots,n-1\}$, let $v_k := a e_1 +xe_{k+1} +ye_{n+1}$, and let $v_{n} := ae_1 +xe_{n} +ye_{2}$ and $v_{n+1} := ae_1 -xe_{2} +ye_{n+1}$.
Since $x^2 +y^2 =r$ now, the $(n+1)$-sized square matrix $B$ with $(v_1,\dots,v_{n+1})$ as its row vectors is a submatrix of the coefficient matrix $A$ of the system of linear equations shown in Eq.~\eqref{eq:a_x_ij}.
More specifically, taking $n=4$ as an example, the square matrix $B$ is as follows:
\begin{equation}\label{eq:matrix_B_xy}
	B=
\left[\begin{array}{c|cccc}
a & x & 0 & 0 & y \\
a & 0 & x & 0 & y \\
a & 0 & 0 & x & y \\
a & y & 0 & x & 0 \\
 \hline
a & -x & 0 & 0 & y
\end{array}\right].
\end{equation}
To prove that $B$ is full-rank, it suffices to show that its determinant $|B| \neq 0$.
To calculate $|B|$, we perform the following row operations on $B$: subtract the last row from the first $4=n$ rows, and arrive at:
\begin{align}\label{eq:det_B_xy}
	|B| &=
\left|\begin{array}{c|cccc}
0 & 2x & 0 & 0 & 0\\
0 & x & x & 0 & 0\\
0 & x & 0 & x & 0\\
0 & x+y & 0 & x & -y\\
 \hline
a & -x & 0 & 0 & y
\end{array}\right|
\\& = (-1)^{n+3} 2a  x^{n-1}y.
\end{align}
Since $a \neq 0$, $x\neq 0$ and $y \neq 0$, we have $|B| \neq 0$, and thus Fact~\ref{fact:linear_independent} holds.

% \section{Proof of Lemma~\ref{lem:rho_to_p}}\label{app:p_rho}

\section{Proof of Lemma~\ref{thm:minimax}}\label{app:minimax}

% \subsection{Another definition of randomized algorithms}

% The definition of a randomized algorithm $\rho: \mathcal{X} \to \mathcal{D}(\mathcal{Y})$ as shown by Eq.~\eqref{eq:rho} is easy to calculate its success probability.
{A randomized algorithm $\rho: \mathcal{X} \to \mathcal{D}(\mathcal{Y})$ as defined by Eq.~\eqref{eq:rho}} is easy to calculate its success probability.
However, it is not convenient for further analysis, as it involves a huge number of $|\mathcal{X}| = |(\F_p^n)^h| = p^{nh}$ probability distributions on the solution space $\mathcal{Y} = \F_p^n$.
% We now show how to reduce the number of probability distributions caused by the randomness within the algorithm from $|\mathcal{X}|$ to $1$, by giving another definition of randomized algorithms as shown by Eq.~\eqref{eq:p} below.
% A deterministic strategy can be characterized by a function $f: \mathcal{X} \to \mathcal{Y}$ that outputs an answer $y \in \mathcal{Y}$ when given a sample combination $x \in \mathcal{X}$.
% As $\mathcal{X}$ and $\mathcal{Y}$ are both finite sets, we can also regard $f$ as a vector $f \in \mathcal{Y}^{|\mathcal{X}|}$.
Alternatively, we can also define a randomized algorithm as
\begin{equation}\label{eq:p}
	\tau \in \mathcal{D}(\mathcal{Y}^{|\mathcal{X}|}),
\end{equation}
where $\mathcal{Y}^{|\mathcal{X}|}$ consists of all the deterministic strategies $f: \mathcal{X} \to \mathcal{Y}$.
The probability of outputting an answer $y \in \mathcal{Y}$ given a sample combination $x \in \mathcal{X}$ can be calculated as follows:
\begin{equation}\label{eq:rho_p}
	\rho(x,y) := \sum_{f \in \mathcal{Y}^{|\mathcal{X}|} : f(x) = y} \tau(f).
\end{equation}
Since the range of $f$ is $\mathcal{Y}$, it is easy to see that $\sum_{y \in \mathcal{Y} } \rho(x,y) = \sum_{f \in \mathcal{Y}^{|\mathcal{X}|} } \tau(f) =1$, and thus $\rho(x,\cdot)$ is indeed a distribution on $\mathcal{Y}$ for any $x \in \mathcal{X}$.
This also means that Eq.~\eqref{eq:rho_p} provides the way to transform from the second definition $\tau \in \mathcal{D}(\mathcal{Y}^{|\mathcal{X}|})$ of a randomized algorithm to the first definition $\rho: \mathcal{X} \to \mathcal{D}(\mathcal{Y})$.

The other direction of transformation from $\rho$ to $\tau$ is not unique, since the second definition $\tau \in \mathcal{D}(\mathcal{Y}^{|\mathcal{X}|})$ has more degrees of freedom than the first definition $\rho: \mathcal{X} \to \mathcal{D}(\mathcal{Y})$, by the inequality $|\mathcal{Y}|^{|\mathcal{X}|} -1 > |\mathcal{X}| \cdot (|\mathcal{Y}| -1)$ when both $|\mathcal{X}|$ and $|\mathcal{Y}|$ are greater than $1$.
In other words, different {$\tau$} can lead to the same {$\rho$}.
Nevertheless, Lemma~\ref{lem:rho_to_p} in the following provides a simple method to recover one such $\tau$ from $\rho$.
% It also implies that to include all the randomized algorithms of the first definition $\rho: \mathcal{X} \to \mathcal{D}(\mathcal{Y})$, it suffices to consider all the randomized algorithms of the second definition $p \in \mathcal{D}(\mathcal{Y}^{|\mathcal{X}|})$.
\begin{Lemma}\label{lem:rho_to_p}
For a randomized algorithm of the first definition $\rho: \mathcal{X} \to \mathcal{D}(\mathcal{Y})$, let $\tau : \mathcal{Y}^{|\mathcal{X}|} \to [0,1]$ be defined as
\begin{equation}\label{eq:p_rho}
	\tau(f) := \prod_{x\in \mathcal{X}} \rho(x,f(x)).
\end{equation}
Then $\tau \in \mathcal{D}(\mathcal{Y}^{|\mathcal{X}|})$ is a randomized algorithm of the second definition as shown by Eq.~\eqref{eq:p}, and it has the same output distribution as $\rho$ in the sense of Eq.~\eqref{eq:rho_p}.
\end{Lemma}
\begin{proof}
We need to prove the following two points.
First, $p$ defined by Eq.~\eqref{eq:p_rho} is indeed a probability distribution on $\mathcal{Y}^{|\mathcal{X}|}$, and thus we need to prove that
\begin{equation}\label{eq:p_rho_facT_1}
	\sum_{f \in \mathcal{Y}^{|\mathcal{X}|}} \tau(f) = 1.
\end{equation}
Second, $\tau$ defined by Eq.~\eqref{eq:p_rho} satisfies the relationship shown in Eq.~\eqref{eq:rho_p}.
Since the summation in Eq.~\eqref{eq:rho_p} is over $f \in \mathcal{Y}^{|\mathcal{X}|}$ such that $f(x) = y$, it suffices to show that
\begin{equation}\label{eq:p_rho_fact_2}
	1 = \sum_{f' \in \mathcal{Y}^{|\mathcal{X}|-1} }
	\prod_{x'\in \mathcal{X}\setminus \{x\} }
	\rho(x',f'(x')).
\end{equation}
From Eqs.~\eqref{eq:p_rho_facT_1} and \eqref{eq:p_rho_fact_2}, it suffices to verify that for any subset $S = \{s_1, \cdots, s_n\} \subseteq \mathcal{X}$, where $1\leq n \leq |\mathcal{X}|$, the following equations hold:
\begin{align}
	1 &= \sum_{ f'(s_1) \in \mathcal{Y} } \rho(s_1, f'(s_1))
	\sum_{ f'(s_2) \in \mathcal{Y} } \rho(s_2, f'(s_2)) \cdots
	\sum_{ f'(s_n) \in \mathcal{Y} } \rho(s_n, f'(s_n)) \label{eq:many_rho_line1}\\
	&= \sum_{ f'(s_1) \in \mathcal{Y} }
	\sum_{ f'(s_2) \in \mathcal{Y} } \cdots
	\sum_{ f'(s_n) \in \mathcal{Y} }
	\rho(s_1, f'(s_1)) \rho(s_2, f'(s_2)) \cdots \rho(s_n, f'(s_n)) \\
	&= \sum_{f' \in \mathcal{Y}^{n}} \prod_{x' \in S} \rho(x', f'(x')),
\end{align}
where Eq.~\eqref{eq:many_rho_line1} follows from the definition of $\rho$ that $\rho(x) \in \mathcal{D}(\mathcal{Y})$ for any $x\in \mathcal{X}$.
\end{proof}
From a higher level of perspective, Eq.~\eqref{eq:rho_p} shows that given the joint distribution $\tau \in \mathcal{D}(\mathcal{Y}^{|\mathcal{X}|})$ of a $|\mathcal{X}|$-dimensional random vector taking values in $\mathcal{Y}^{|\mathcal{X}|}$, we can compute its marginal distributions $\rho: \mathcal{X} \to \mathcal{D}(\mathcal{Y})$,
whereas Eq.~\eqref{eq:p_rho} shows that when the marginal distributions $\{\rho(x) : x\in\mathcal{X} \}$ are given, we can recover a joint distribution that has the same marginal distributions by simply taking their products.

The RHS of Eq.~\eqref{eq:minimax_sample} in Lemma~\ref{thm:minimax} involves the first definition $\rho: \mathcal{X} \to \mathcal{D}(\mathcal{Y})$ of randomized algorithms, which is not easy to deal with.
We now show that it can be converted to the second form $\tau \in \mathcal{D}(\mathcal{Y}^{|\mathcal{X}|})$ of randomized algorithms:
\begin{equation}\label{eq:minimax_p_rho}
	\max_{\tau \in \mathcal{D}(\mathcal{Y}^{|\mathcal{X}|})} \min_{y \in \mathcal{Y}} \mathop{\mathbb{E}}\limits_{F \sim \tau} [\Pr_{X \sim \sigma(y)} (F(X)=y)]
	=
	\max_{\rho\in [\mathcal{X} \to \mathcal{D}(\mathcal{Y})]} \min_{y\in \mathcal{Y}} \mathop{\mathbb{E}}\limits_{X\sim \sigma(y)} [\rho(X,y)]
\end{equation}

\begin{proof}[Proof of Eq.~\eqref{eq:minimax_p_rho}]
% We will prove from the two ``$\geq$'' and ``$\leq$" directions.
{
We will prove from two directions: the LHS of Eq.~\eqref{eq:minimax_p_rho} is greater than or equal to its RHS,
and the LHS of Eq.~\eqref{eq:minimax_p_rho} is less than or equal to its RHS.
}
To prove the general inequality $\max_{a\in A } \min_{b \in B} \{F_{a,b}\} \leq \max_{c\in C} \min_{b\in B} \{G_{c,b}\}$, it suffices to show that for any $a\in A$ there is a $c\in C$ such that $\min_{b \in B} \{F_{a,b}\} \leq \min_{b\in B} \{G_{c,b}\}$, {for which it} suffices to show that $F_{a,b} = G_{c,b}$ for any $b\in B$.

We first prove that {the LHS of Eq.~\eqref{eq:minimax_p_rho} is greater than or equal to its RHS.}
As explained above, it suffices to show that for any randomized algorithm $\rho \in [\mathcal{X} \to \mathcal{D}(\mathcal{Y})]$ of the first definition, the randomized algorithm $\tau \in \mathcal{D}(\mathcal{Y}^{|\mathcal{X}|})$ of the second definition defined as $\tau(f) := \prod_{x\in \mathcal{X}} \rho(x,f(x))$ by Eq.~\eqref{eq:p_rho} satisfies:
\begin{equation}\label{eq:E_p_rho}
	\mathop{\mathbb{E}}\limits_{F \sim \tau} [\Pr_{X \sim \sigma(y)} (F(X)=y)] = \mathop{\mathbb{E}}\limits_{X\sim \sigma(y)} [\rho(X,y)]
\end{equation}
for any solution $y\in \mathcal{Y}$. Let $\delta[A] \in \{1,0\}$ be the boolean function that indicates whether event $A$ happens or not.
Then Eq.~\eqref{eq:E_p_rho} can be verified as follows:
\begin{align}
	\mathop{\mathbb{E}}\limits_{F \sim \tau} [\Pr_{X \sim \sigma(y)} (F(X)=y)] 
	&= \sum_{f \in \mathcal{Y}^{|\mathcal{X}|}} \tau(f) \sum_{x\in \mathcal{X}} \sigma(y,x) \cdot \delta[f(x)=y] \label{eq:sample_pr_line_1}\\
	&= \sum_{x\in \mathcal{X}} \sigma(y,x) \sum_{f \in \mathcal{Y}^{|\mathcal{X}|}} \tau(f) \cdot \delta[f(x)=y] \\
	&= \sum_{x\in \mathcal{X}} \sigma(y,x) \rho(x,y) \label{eq:sample_pr_line_3} \\
	&= \mathop{\mathbb{E}}\limits_{X\sim \sigma(y)} [\rho(X,y)], \label{eq:sample_pr_line_4}
\end{align}
where Eq.~\eqref{eq:sample_pr_line_3} follows from the fact that $\tau(f) = \prod_{x\in \mathcal{X}} \rho(x,f(x))$ satisfies $\sum_{f \in \mathcal{Y}^{|\mathcal{X}|} : f(x) = y} \tau(f) = \rho(x,y)$ by Lemma~\ref{lem:rho_to_p}.

We now prove that {the LHS of Eq.~\eqref{eq:minimax_p_rho} is less than or equal to its RHS.}
As explained above, it suffices to show that for any randomized algorithm $\tau \in \mathcal{D}(\mathcal{Y}^{|\mathcal{X}|})$ of the second definition, the randomized algorithm $\rho \in [\mathcal{X} \to \mathcal{D}(\mathcal{Y})]$ of the first definition defined as $\rho(x,y) := \sum_{f \in \mathcal{Y}^{|\mathcal{X}|}} \tau(f) \cdot \delta[f(x)=y]$ by Eq.~\eqref{eq:rho_p} satisfies:
\begin{equation}
	\mathop{\mathbb{E}}\limits_{X\sim \sigma(y)} [\rho(X,y)] = \mathop{\mathbb{E}}\limits_{F \sim \tau} [\Pr_{X \sim \sigma(y)} (F(X)=y)],
\end{equation}
for any solution $y\in \mathcal{Y}$.
This can be verified by simply carrying out the calculation shown in Eqs.~\eqref{eq:sample_pr_line_1} to \eqref{eq:sample_pr_line_4} in the reverse order.
\end{proof}

From Eq.~\eqref{eq:minimax_p_rho}, to prove Lemma~\ref{thm:minimax}, it suffices to show:
\begin{equation}\label{eq:minimax_f_p}
	\min_{\mu \in \mathcal{D}(\mathcal{Y})}
	\max_{f \in \mathcal{Y}^{|\mathcal{X}|}}
	\mathop{\mathbb{E}}\limits_{Y\sim \mu}
	[\Pr_{X \sim \sigma(Y)} (f(X)=Y)]
	= \max_{\tau \in \mathcal{D}(\mathcal{Y}^{|\mathcal{X}|})}
	\min_{y \in \mathcal{Y}}
	\mathop{\mathbb{E}}\limits_{F \sim \tau} [\Pr_{X \sim \sigma(y)} (F(X)=y)].
\end{equation}
Consider the LHS of the above Eq.~\eqref{eq:minimax_f_p}.
Denote by
\begin{equation}\label{eq:P_element}
	P(f,y) := \Pr_{X \sim \sigma(y)} (f(X)=y),
\end{equation}
the success probability of the deterministic strategy $f$ under the random sample combination $X \sim \sigma(y)$.
Let us regard $P$ as a matrix with $|\mathcal{Y}|^{|\mathcal{X}|}$ rows and $|\mathcal{Y}|$ columns, whose elements $P(f,y)$ are defined by Eq.~\eqref{eq:minimax_f_p}.
For a finite set $\mathcal{Z}$, denote by $e_z$ the $|\mathcal{Z}|$-dimensional column vector with $1$ in the $z$ component and $0$ in the other components.
Let us also regard a probability distribution $\sigma \in \mathcal{D}(\mathcal{Z})$ as a $|\mathcal{Z}|$-dimensional column vector. Then Eq.~\eqref{eq:minimax_f_p} reads:
\begin{equation}\label{eq:min_max_sample}
	\min_{\mu \in \mathcal{D}(\mathcal{Y})} \max_{f \in \mathcal{Y}^{|\mathcal{X}|}} e_f^T P \mu
	=
	\max_{\tau \in \mathcal{D}(\mathcal{Y}^{|\mathcal{X}|})} \min_{y \in \mathcal{Y}} \tau^T P e_y. 
\end{equation}

The above Eq.~\eqref{eq:min_max_sample} bears similarity with the following Eq.~\eqref{eq:min_max_von} known as von Neumann's minimax theorem~\cite{Neumann1928}.
\begin{equation}\label{eq:min_max_von}
	\min_{\sigma_2 \in \mathcal{D}(A_2)} \max_{\sigma_1 \in \mathcal{D}(A_1)} \sigma_1^T u \sigma_2
	= \max_{\sigma_1 \in \mathcal{D}(A_1)} \min_{\sigma_2 \in \mathcal{D}(A_2)} \sigma_1^T u \sigma_2,
\end{equation}
where $A_1$ and $A_2$ are finite sets, and $u$ is a real matrix with $|A_1|$ rows and $|A_2|$ columns.
The above Eq.~\eqref{eq:min_max_von} was first proved by von Neumann~\cite{Neumann1928} in 1928, but the most popular proof may be the one based on separating hyperplane theorem appeared in von Neumann and Morgenstern’s famous book on game theory~\cite{game_theory_1945} published in 1944.
The proof by Loomis~\cite{Loomis_1946} in 1946 based on mathematical induction on $|A_1|+|A_2|$ is probably the simplest, and Weinstein~\cite{Weinstein_2022} has simplified Loomis's proof to only half a page.

Comparing Eq.~\eqref{eq:min_max_sample} with Eq.~\eqref{eq:min_max_von}, it suffices to show that:
\begin{equation}\label{eq:max_sigma_1}
	\max_{\sigma_1 \in \mathcal{D}(A_1)} \sigma_1^T u \sigma_2 = \max_{a_1\in A_1} e_{a_1}^T u \sigma_2,
\end{equation}
for any $\sigma_2 \in \mathcal{D}(A_2)$; and that:
\begin{equation}\label{eq:min_sigma_2}
	\min_{\sigma_2 \in \mathcal{D}(A_2)} \sigma_1^T u \sigma_2 = \min_{a_2 \in A_2} \sigma_1^T u e_{a_2},
\end{equation}
for any $\sigma_1 \in \mathcal{D}(A_1)$.
We now prove Eq.~\eqref{eq:max_sigma_1}, and Eq.~\eqref{eq:min_sigma_2} can be shown with similar arguments.
When $\sigma_1$ is fixed, $\sigma_1^Tu$ is a $|A_2|$-dimensional row vector and we {denote} it by $v^T$.
Assume that its minimum component is $v(i_0)$.
Then for any $\sigma_2 \in \mathcal{D}(A_2)$, $v^T \sigma_2 = \sum_{i\in A_2} v(i) \sigma(i) \geq v(i_0) \sum_{i\in A_2}\sigma(i) = v(i_0)$.
Note also that $v^T \sigma_2 =v(i_0)$ when $\sigma_2 = e_{i_0}$.
Therefore, $ \min_{\sigma_2 \in \mathcal{D}(A_2)} \sigma_1^T u \sigma_2 = v(i_0)$, where $v(i_0)$ by definition is equal to $\min_{a_2 \in A_2} \sigma_1^T u e_{a_2}$.

Plugging Eqs.~\eqref{eq:max_sigma_1} and \eqref{eq:min_sigma_2} back to Eq.~\eqref{eq:min_max_von}, we have proven Eq.~\eqref{eq:min_max_sample}, which is the matrix form of Eq.~\eqref{eq:minimax_f_p}.
Therefore, we have proven Lemma~\ref{thm:minimax}.

\section{Lower bound on the quantum algorithm's success probability}\label{app:p_q_lower}

In this section, we give a detailed analysis of the lower bound on the success probability of the quantum algorithm presented above Lemma~\ref{lem:CTQW_prob}, according to which we denote by $P_t = \Omega(1/\log p)$ the lower bound on the probability of measuring the center $t$, and by $P_x = O(1/p^{-(n-1)})$ the upper bound on the probability of obtaining any other point $y \in \F_p^n \setminus \{t\}$.

\begin{Lemma}\label{lem:p_q_lower}
Let $c>0$ be a constant such that the number $T = c/P_t$ of iterations (steps 1 and 2 in the quantum algorithm) is an integer greater than $2$.
Then the success probability $P_s$ of obtaining the center $t$ has the following lower bound:
\begin{equation}\label{eq:P_s}
	P_s \geq 1 - e^{-c} -\frac{c P_x}{1-P_t} e^{-c} -\frac{p^n c^2 P_x^2}{2P_t^2}.
\end{equation}
\end{Lemma}

Lemma~\ref{lem:p_q_lower} implies that the quantum algorithm will obtain the center $t$ with high probability when $p$ is fixed and $n$ is large enough.
The reason is as follows.
According to Eq.~\eqref{eq:sphere_size}, the size of a sphere has an exponentially large lower bound, i.e. $s_r \geq p^{n-1} -p^{(n-1)/2}$, and hence $P_x \leq 2/p^{n-1}$.
Substituting it into Eq.~\eqref{eq:P_s}, we have:
\begin{align}
	P_s & \geq 1 -e^{-c} -\frac{2p c e^{-c}}{1-P_t}p^{-n} - \frac{2p^2 c^2 }{P_t^2} p^{-n} \\
	& \geq 1 -e^{-c} -O(p^{-n}).
\end{align}
Since we regard $p$ and $P_t = \Omega(1/\log p)$ as constants, the success probability of the quantum algorithm $P_s \geq 1 -e^{-c} -o(1)$ has a constant lower bound when $n$ is large enough.

Before giving the proof of Lemma~\ref{lem:p_q_lower}, we present an example showcasing a significant quantum-classical separation.
According to the numerical simulations in Appendix~\ref{sec:numerical}, the maximum success probability of obtaining the center $t$ after a period of CTQW is lower bounded by $P_t =0.4$ when $p=3, n=12$, and $r=1$.
Therefore, the probability of obtaining any other point is upper bounded by $P_x = (1-P_t)/(p^{n-1}-p^{(n-1)/2})$.
Set the number of iterations as $T=3$, which is also the number of quantum samples used, then the constant $c = TP_t =1.2$.

Substituting all the parameters above into Eq.~\eqref{eq:P_s}, it can be calculated numerically that $P_s \geq 0.698$.
On the other hand, the upper bound on the success probability of any classical randomized algorithm is $P_c = p^{-(1-2T/n)n} \leq 0.0014$ by Theorem~\ref{thm:classical_lower}.
Therefore in this example, the success probability of the quantum algorithm is at least $P_s/P_c \geq 509$ times that of any classical algorithm, provided with the same number $T=3$ of quantum or classical samples.

\begin{proof}[Proof of Lemma~\ref{lem:p_q_lower}.]
Consider the following event: among the $T$ measurement outcomes $x= (x_1,\cdots,x_T)$, the number of occurences of the center $t$ denoted by $C(x,t)$ satisfies $C(x,t) \geq 2$, and the number of occurences of any other point $y \in \F_p^n \setminus \{t\}$ satisfies $C(x,y) \leq 1$.
If this event happens, then the algorithm succeeds.
Therefore, the lower bound on the success probability $P_s$ can be calculated as follows:
\begin{align}
	P_s & \geq \Pr[C(x,t) \geq 2 \cap \bigcap_{ y \in \F_p^n \setminus \{t\} } C(x,y) \leq 1 ] \\
    & \geq 1- \Pr[C(x,t) =0] -\Pr[C(x,t) =1] -\Pr[\bigcup_{ y \in \F_p^n \setminus \{t\} } C(x,y) \geq 2] \label{eq:P_s_line_2} \\
    & \geq 1 - (1-P_t)^T -T(1-P_t)^{T-1}P_x - (p^n-1)\binom{T}{2} P_x^2 \label{eq:P_s_line_3} \\
    & \geq 1 - e^{-c} -\frac{c P_x}{1-P_t} e^{-c} -\frac{p^n c^2 P_x^2}{2P_t^2}. \label{eq:P_s_line_4}
\end{align}
In Eq.~\eqref{eq:P_s_line_2} we used the basic formula $\Pr(A) = 1- \Pr(\neg A)$ where $A$ is an event, and the union bound $\Pr(A \cup B) \leq \Pr(A) + \Pr(B)$, as well as De Morgan's law of events $\neg (A \cap B) = \neg A \cup \neg B$.
The last term of Eq.~\eqref{eq:P_s_line_3} uses the equivalent meaning of the event $\cup_{ y \in \F_p^n \setminus \{t\} } C(x,y) \geq 2$, i.e. there is a pair of subscripts $i \neq j$ in the $T$ measurement outcomes $x= (x_1,\cdots,x_T)$ such that $x_i =x_j =y \in \F_p^n \setminus \{t\}$.
This can be formally expressed as $\cup_{ y \in \F_p^n \setminus \{t\} } \cup_{ \{i \neq j\} \subseteq [T] } [ (x_i=y) \cap (x_j=y) ] $, where $[T] := \{1,\dots,T\}$.
In Eq.~\eqref{eq:P_s_line_4} we used the fact that $(1-x)^{1/x} \leq e^{-1}$ for $x \in (0,1)$.
\end{proof}

\section{Invariant supspace of the CTQW}\label{app:CTQW}
We will prove the following Lemma~\ref{lem:invariant}, used in the proof of Lemma~\ref{lem:CTQW_prob} and also in obtaining the reduced matrix $M_A$ in Eq.~\eqref{eq:M_A}.
Lemma~\ref{lem:invariant} is extracted from~\cite[Lemma 4]{nonlinear}, but with a more detailed proof for the convenience of the reader.
Lemma~\ref{lem:invariant} implies that $\mathcal{H}_0$ is an invariant subspace of the CTQW $e^{i \bar{A} T_0} = \sum_{k=0}^{\infty} \frac{(iT_0)^k}{k!} \bar{A}^k$, since $\bar{A} = A -\lambda_0 \ket{\widetilde{0}} \bra{\widetilde{0}}$, and $\sqrt{p^n}\ket{\tilde{0}} =\sum_{x\in\F_p^n} \ket{x} =\ket{t} + \sum_{r'=0}^{p-1} \sqrt{s_{r'}} \ket{S_{r'}+t} \in \mathcal{H}_0$.

\begin{Lemma}\label{lem:invariant}
The adjacency matrix $A$ has a $(p+1)$-dimensional invariant subspace $\mathcal{H}_0$ spanned by the following orthonormal basis:
\begin{equation}
        \mathcal{B}_0 = \{ \ket{t}, \ket{S_0+t}, \ket{S_1+t},\cdots,\ket{S_{p-1}+t} \}.
    \end{equation}
Specifically, we have
\begin{align}
A\ket{t} &= \sqrt{s_r}\ket{S_r+t}, \label{eq:At}\\
A\ket{S_{r'}+t} &= c_{r'} \ket{t} +\sum_{r''=0}^{p-1} c(r'',r') \ket{S_{r''}+t}, \label{eq:Srt_B_0}
\end{align}
where $c(r'',r') := \sqrt{\frac{s_{r''}}{s_{r'}}} \left| S_r \cap (S_{r'}+v_0'') \right|$ for arbitrary $v_0''\in S_{r''}$, and $c_{r'} := \delta_{r',r}\sqrt{s_r}$.
\end{Lemma}

\begin{proof}
Eq.~\eqref{eq:At} follows from the definition that $A$ maps any $x\in \F_p^n$ to $\{y\in {\F}_p^n \mid l(y-x)=r\} = S_r+x$.
To prove Eq.~\eqref{eq:Srt_B_0}, we consider any point $v''+t$ on $S_{r''}+t$, where $v'' \in S_{r''}$.
We calculate
\begin{align}
&\bra{v''+t} A \ket{S_{r'}+t} \\
&= \frac{1}{\sqrt{s_{r'}}} \sum_{v\in S_r} \sum_{v'\in S_{r'}} \braket{v''+t | v-v'+t} \label{eq:S_r_expand}\\
&= \frac{1}{\sqrt{s_{r'}}} \left| S_r \cap (S_{r'}+v'')\right|. \label{eq:intersection}
\end{align}
We used $\ket{S_{r'} +t} = \frac{1}{\sqrt{s_{r'}}} \sum_{v'\in S_{r'}} \ket{-v'+t}$ in Eq.~\eqref{eq:S_r_expand}, and Eq.~\eqref{eq:intersection} follows from the fact that the condition $v''+t=v-v'+t$ is equivalent to $v=v'+v''$, where $v\in S_r$ and $v'\in S_{r'}$.

We will later show that $\left| S_r \cap (S_{r'}+v'')\right|$ is the same for any $v'' \in S_{r''}$.
Thus we have
\begin{align}
&A\ket{S_{r'}+t} \\
&= c_{r'}\ket{t} +\sum_{r''=0}^{p-1}\sum_{v''\in S_{r''}}\bra{v''+t} A \ket{S_{r'}+t}\ket{v''+t} \label{eq:ASRT_1}\\ 
&= c_{r'}\ket{t}
%\nonumber\\&\quad 
+\sum_{r''=0}^{p-1} \frac{1}{\sqrt{s_{r'}}} \left| S_r \cap (S_{r'}+v_0'')\right| \sum_{v''\in S_{r''}} \ket{v''+t} \label{eq:ASRt_2}\\
&= c_{r'}\ket{t} 
+\sum_{r''=0}^{p-1} c(r'',r') \ket{S_r''+t}.
\end{align}
We used $\bra{t}A\ket{S_{r'}+t} = \bra{S_{r'}+t} A \ket{t} = \delta_{r',r}\sqrt{s_1}$ in Eq.~\eqref{eq:ASRT_1}, where the first equality follows from $A$ being symmetric.
In Eq.~\eqref{eq:ASRt_2}, $v_0''\in S_{r''}$ is arbitrary.

We now show that $\left| S_r \cap (S_{r'}+v'')\right|$ is the same for any $v'' \in S_{r''}$, or equivalently, $| S_r \cap (S_{r'}+x)| = |S_r \cap (S_{r'}+z)|$ for any $x\in S_{r''}$ and $z\in S_{r''}$.
We will later construct an isometry $\tau$ of $\F_p^n$ such that $z = \tau(x)$.
As $\tau$ is a distance-preserving bijection, we have $\tau(S_r) = S_r$ and $\tau(S_{r'})=S_{r'}$, and thus $\tau(S_r\cap(S_{r'}+x)) = S_r\cap(S_{r'}+z)$, which implies $| S_r \cap (S_{r'}+x)| = |S_r \cap (S_{r'}+z)|$.
The isometry $\tau$ can be constructed by extending the isometry $\tau' : ax \mapsto az$ that maps subspace $U=\{ a x :a\in \F_p\}$ to subspace $W = \{ az : a\in \F_p \}$ (note that $U$ and $W$ are both $1$-dimensional subspace, since $x$ and $z$ are both nonzero points), to an isometry $\tau$ of the whole space $\F_p^n$ using Witt's Lemma~\cite[Section 20]{Aschbacher_2000}.
\end{proof}

\end{appendices}

\end{document}